\documentclass[12pt]{article}
\usepackage{amssymb}
\usepackage{amsmath}
\usepackage{amsthm}
\usepackage{latexsym}
\usepackage{mathdots}
\usepackage{graphicx}
\usepackage{lineno}
\usepackage[all]{xy}
\newtheorem{theo}{Theorem}[section]
\newtheorem{prop}[theo]{Proposition}
\newtheorem{cor}[theo]{Corollary}
\newtheorem{lemma}[theo]{Lemma}
\theoremstyle{definition}
\newtheorem{defi}[theo]{Definition}
\newtheorem{exa}[theo]{Example}
\newtheorem{rem}[theo]{Remark}

\numberwithin{equation}{section}

\newcommand{\N}{{\mathbb N}}
\newcommand{\F}{{\mathbb F}}

\newcommand{\cC}{{\mathcal C}}
\newcommand{\cD}{{\mathcal D}}

\newcommand{\cM}{{\mathcal M}}

\newcommand{\cT}{{\mathcal T}}

\newcommand{\cI}{{\mathcal I}}
\newcommand{\cR}{{\mathcal R}}
\newcommand{\cS}{{\mathcal S}}

\newcommand{\T}{\mbox{$\!^{\sf T}$}}

\newcommand{\ideal}[1]{\mbox{$({#1})$}}
\newcommand{\lideal}[1]{\mbox{$^{\bullet}({#1})$}}
\newcommand{\rideal}[1]{\mbox{$({#1})^{\bullet}$}}
\newcommand{\Aut}{\mbox{\rm Aut}}
\newcommand{\Fix}{\mbox{\rm Fix}}

\newcommand{\vv}{{\frak{v}}}
\newcommand{\ppa}{{\frak{p}_a}}
\newcommand{\vva}{{\frak{v}_a}}
\newcommand{\coset}[1]{\overline{#1}}

\newcommand{\rmid}{\mbox{$\,|_{r}\,$}}

\newcommand{\rank}{{\rm rk}}
\newcommand{\gcrd}{{\rm gcrd}}
\newcommand{\lclm}{{\rm lclm}}

\newcommand{\im}{{\rm im}\,}

\newcommand{\CircMatat}{\mbox{$M_a^{\theta}$}}
\newcommand{\CircMat}[1]{\mbox{$M_{#1}^{\theta}$}}
\newcounter{alp}
\newcounter{ara}
\newcounter{rom}
\newenvironment{romanlist}{\begin{list}{(\roman{rom})\hfill}{\usecounter{rom}
     \topsep0ex \labelwidth.7cm \leftmargin.7cm \labelsep0cm
     \rightmargin0cm \parsep0ex \itemsep.4ex
     \partopsep1ex}}{\end{list}}
\newenvironment{alphalist}{\begin{list}{(\alph{alp})\hfill}{\usecounter{alp}
     \topsep0ex \labelwidth.7cm \leftmargin.7cm \labelsep0cm
     \rightmargin0cm \parsep0ex \itemsep.5ex
     \partopsep0ex}}{\end{list}}
\newenvironment{arabiclist}{\begin{list}{(\arabic{ara})\hfill}{\usecounter{ara}
     \topsep0ex \labelwidth.7cm \leftmargin.7cm \labelsep0cm
     \rightmargin0cm \parsep0ex \itemsep.5ex
     \partopsep0ex}}{\end{list}}

\begin{document}
\title{A Circulant Approach to Skew-Constacyclic Codes}
\date{\today}
\author{Neville Fogarty and Heide Gluesing-Luerssen\footnote{HGL was partially supported by the National
Science Foundation Grant DMS-1210061.
Both authors are with the Department of Mathematics, University of Kentucky, Lexington KY 40506-0027, USA;
\{neville.fogarty,heide.gl\}@uky.edu.}}

\maketitle

{\bf Abstract:}
We introduce circulant matrices that capture the structure of a skew-poly\-nomial ring $\F[x;\theta]$
modulo the left ideal generated by a polynomial of the type~$x^n-a$.
This allows us to develop an approach to skew-constacyclic codes based on such circulants.
Properties of these circulants are derived, and in particular it is shown that the transpose of a
certain circulant is a circulant again.
This recovers the well-known result that the dual of a skew-constacyclic code is a constacyclic code again.
Special attention is paid to the case where $x^n-a$ is central.



\section{Introduction}\label{S-Intro}
Cyclic block codes form the most powerful class of linear block codes due to their inherent algebraic structure
which allows the design of codes with large distance and efficient decoding algorithms.
In recent years the notion of cyclicity has been generalized to skew-cyclicity, mainly in the work by
Boucher/Ulmer and coworkers, see \cite{BGU07,BoUl09,CLU09,BoUl11,BoUl14}, but also by
Abualrub et al.~\cite{AGAS10}, Matsuoka~\cite{Mat10}, and Gao et al.~\cite{GSF13}.

These codes are defined and studied with the aid of skew-polynomial rings.
These are rings of the form~$\F[x;\theta]$ or even~$\F[x;\theta,\delta]$ with an automorphism~$\theta$ and a
$\theta$-derivation~$\delta$, and where~$\theta$ and~$\delta$ describe the relation between $ax$ and~$xa$ for
coefficients~$a\in\F$.
They were introduced by Ore~\cite{Ore33} in 1933.
It is interesting to observe that, beyond the area of skew-constacyclic codes, skew-polynomial rings over finite fields
have gained considerable attention in recent years in coding theory, shift-register synthesis, and cryptography;
see for instance \cite{LMK14,SJB11,SiBo14,BGGU10,Zha10,Wu14}.

In the papers mentioned in the first paragraph, most notably~\cite{BoUl09,BoUl11,BoUl14}, an algebraic theory of skew-constacyclic codes
has been developed.
It generalizes -- to a large extent -- the classical algebraic theory of cyclic codes.
For instance, a central result in~\cite{BoUl11} is that the dual code of a skew-constacyclic code is again
skew-constacyclic.

In~\cite{BGU07,CLU09} the authors present skew-constacyclic codes whose distance improves upon the largest
distance that was known at that time for codes with the same parameters $(q,n,k)$.
In~\cite{AGAS10} the same is done using skew quasi-cyclic codes.
In~\cite{BoUl14a} some self-dual skew-constacyclic codes are found that have better distance than previously known
self-dual codes with the same parameters.
All of this suggests that the class of skew-constacyclic codes has some promising potential.
One reason for this may be that in skew-polynomial rings $\F_q[x;\theta]$, polynomials do not
factor uniquely into irreducibles and therefore often have a large number of (right) divisors.
As a consequence, one obtains plenty of skew-constacyclic codes.
The latter are defined as the submodules generated by right divisors of some $x^n-a$ in the left module
of skew-polynomials $\cR:=\F[x;\theta]$ modulo the left ideal generated by~$x^n-a$.

In this paper we will develop an approach to skew-constacyclic codes with the aid of suitably defined circulant matrices,
thereby rediscovering the above duality result.

A circulant description of classical cyclic codes is well known (see for instance~\cite[p.~501]{MS77}).
In that case, the circulant associated with a polynomial~$g$ is a square matrix whose $i$-th row contains the
coefficients of $x^ig$ modulo $x^n-1$ for $i=0,\ldots,n-1$.
In our context, circulants are matrices where the rows are the lists of left coefficients of the left multiples
$x^ig\in\cR$ modulo~$\cR(x^n-a)$.
We will show that if~$g$ is a right divisor of $x^n-a$, then the transpose of its circulant is, up to reordering and rescaling of its rows,
the circulant of a right divisor of $x^n-c$ for a particular constant~$c=c(a,g)$.
Since the row space of the circulant is the skew-constacyclic code generated by~$g$, this result will recover the duality
theorem proven by Boucher/Ulmer in~\cite{BoUl11}.

Furthermore, with the aid of a particular product formula for circulants we obtain anti-isomorphisms between
the lattice of right divisors of~$x^n-a$, the lattice of right divisors of $x^n-a^{-1}$, the lattice of skew-constacyclic
codes in~$\F^n$ and the lattice of dual codes.
These results can be derived despite the fact that the theory of circulants does not entirely generalize from the classical case
to the skew-polynomial case.
For instance, in general products of circulants are not circulants and neither are their transposes.
Only for right divisors of~$x^n-a$ can the necessary relations be obtained.

Finally, special attention will be paid to the case where the left ideal $\cR(x^n-a)$ is a two-sided ideal.
In this case the circulants form a subring of $\F^{n\times n}$ which is isomorphic to the quotient ring  $\cR/\cR(x^n-a)$.
As a consequence, the theory nicely generalizes the commutative case, as it can be found in, ~e.g.,~\cite[p.~501]{MS77}.
This is in stark contrast to the general case, in which general circulants satisfy only few properties, as we pointed out above.

\section{Preliminaries}\label{S-Prelim}

Let~$\F$ be a finite field and $\theta\in\Aut(\F)$, that is,~$\theta$ is an automorphism of~$\F$.
We consider the skew polynomial ring $\cR:=\F[x;\theta]$, which is defined as the set
$\{\sum_{i=0}^n a_i x^i\mid n\in\N_0,\,a_i\in\F\}$ endowed with the usual addition, and where multiplication
is given by
\[
   xa=\theta(a)x\text{ for all }a\in\F
\]
together with the laws of associativity and distributivity.
Then~$\cR$ is a ring with identity which is non-commutative unless $\theta=\text{id}_{\F}$.
Following Boucher/Ulmer~\cite{BoUl09}, we call~$\cR$ a \emph{skew-polynomial ring of automorphism type}.
Despite the non-commutativity, the ring is very similar to ordinary polynomial rings over fields.
Some well-known properties are summarized below.
Note that the degree of a polynomial $f\in\cR$, denoted by $\deg(f)$, does not depend on the side where we collect
the coefficients of~$f$ since~$\theta$ is an automorphism.
We also define $\deg(0)=-\infty$.
Then we have the usual degree formulas, and in particular~$\cR$ is a domain.
It is easy to see that the center of~$\cR$ is given by
\begin{equation}\label{e-center}
   Z(\cR)=\widehat{\F}[x^m],\ \text{ where }|\theta|=m,
\end{equation}
and $\widehat{\F}:=\Fix_{\F}(\theta)$ is the fixed field of~$\theta$.

%%%%%%%%%%%%%%%%%%%%%%%%%%
\begin{rem}[\mbox{\cite{Ore33}}]\label{R-PropR}
$\cR$ is a left Euclidean domain and a right Euclidean domain.
More precisely, we have the following.
\begin{alphalist}
\item (Right division with remainder) For all $f,\,g\in\cR$ with $g\neq0$ there exist unique polynomials $s,r\in\cR$ such that
      $f=sg+r$ and $\deg(r)<\deg(g)$.
      If $r=0$, then~$g$ is a \emph{right divisor} of~$f$, denoted by $g\rmid f$.
\item For any two polynomials $f_1,\,f_2\in\cR$, not both zero, there exists a unique monic polynomial~$d\in\cR$
      such that $d\rmid f_1,\ d\rmid f_2$ and such that whenever $h\in\cR$ satisfies
      $h\rmid f_1$ and $h\rmid f_2$ then $h\rmid d$.
      The polynomial~$d$ is called the \emph{greatest common right divisor} of $f_1$ and~$f_2$, denoted by $\gcrd(f_1,f_2)$.
      It satisfies a \emph{right Bezout identity}, that is,
      \[
          d=uf_1+vf_2\ \text{ for some }u,\,v\in\cR.
      \]
      We may choose $u,\,v$ such that $\deg(u)<\deg(f_2)$ and, consequently, $\deg(v)<\deg(f_1)$; see~\cite[Sec.~2]{Gie98}.
\item For any two nonzero polynomials $f_1,\,f_2\in\cR$, there exists a unique monic polynomial~$\ell\in\cR$
      such that $f_i\rmid\ell,\,i=1,2,$ and such that whenever $h\in\cR$ satisfies
      $f_i\rmid h,\,i=1,2,$ then $\ell\rmid h$.
      The polynomial~$\ell$ is called the \emph{least common left multiple} of $f_1$ and~$f_2$, denoted by $\lclm(f_1,f_2)$.
      Moreover, we have $\ell=uf_1=vf_2$ for some $u,\,v\in\cR$ with $\deg(u)\leq\deg(f_2)$ and $\deg(v)\leq\deg(f_1)$;
      this follows from \cite[Thm.~8 and Eq.~(24)]{Ore33}.
\item For all nonzero $f_1,\,f_2\in\cR$
      \[
           \deg(\gcrd(f_1,f_2))+\deg(\lclm(f_1,f_2))=\deg(f_1)+\deg(f_2).
      \]
\end{alphalist}
Analogous statements hold true for the left hand side.
\end{rem}
%%%%%%%%%%%%%%%%%%%%%%%%%

Let now $a\in\F^*:=\F\backslash\{0\}$ and $n\in\N$.
Throughout this paper we will be concerned with the quotient module
\[
       \cS_a:=\cR/\lideal{x^n-a},
\]
where $\lideal{x^n-a}:=\cR(x^n-a)$ denotes the principal left ideal generated by~$x^n-a$.
Note that in general~$\cS_a$ is not a ring, but simply a left~$\cR$-module.
This naturally induces a left $\F$-vector space structure as well.

The coset $f+\cR(x^n-a)$ of $f\in\cR$ will be denoted by $\coset{f}$.
The left $\cR$-module structure implies $t\,\coset{f}=\coset{tf}$ for any $t,\,f\in\cR$.
From right division with remainder it is clear that every coset in~$\cS_a$ has a unique representative
of degree less than~$n$.

Occasionally we will pay special attention to the case where~$\cS_a$ is a ring.
%%%%%%%%%%%%%%%%%%%%%%%%%%%%%
\begin{rem}\label{R-twosided}
An element~$f\in\cR$ is called \emph{two-sided} if $\cR f=f\cR$.
In this case the left ideal $\cR f$ is even two-sided and thus $\cR/\cR f$ is a ring.
It is not hard to see \cite[Thm.~1.1.22]{Jac96} that the two-sided elements of~$\cR$ are exactly the skew-polynomials of the form
$cx^tf$, where~$c\in\F$ and $t\in\N_0$, and~$f$ is in the center~$Z(\cR)$.
In particular, a polynomial of the form $x^n-a$, where~$a\neq0$, is two-sided if and only if it is central and this is the case if and only if~$|\theta|$ divides~$n$ and $a\in\Fix_{\F}(\theta)$.
Only in this case is the module $\cS_a=\cR/\lideal{x^n-a}$ a ring.
\end{rem}
%%%%%%%%%%%%%%%%%%%%%%%%%%%

Let us return to the general case.
The module~$\cS_a$ is the skew-constacyclic analogue of the quotient ring $\F[x]/\ideal{x^n-1}$ for cyclic codes or,
more generally, of $\F[x]/\ideal{x^n-a}$ for constacyclic codes.
We have the left $\F$-linear isomorphism
\begin{equation}\label{e-ppa}
  \ppa:\,\F^n\longrightarrow \cS_a,\ (c_0,\ldots,c_{n-1})\longmapsto \coset{\sum_{i=0}^{n-1} c_i x^i}.
\end{equation}
It is crucial that the coefficients~$c_i$ appear on the left of~$x$, because only this turns~$\ppa$ into an
isomorphism of (left) $\F$-vector spaces.
This map will relate codes in~$\F^n$ to submodules in~$\cS_a$.
We set
\begin{equation}\label{e-vva}
  \vva:=\ppa^{-1}.
\end{equation}

The following facts about submodules of~$\cS_a$ are straightforward generalizations of the commutative case and
are proven in exactly the same way (with the aid of Remark~\ref{R-PropR}).
Just as for left ideals we use the notation $\lideal{\coset{g}}$ for the left submodule of~$\cS_a$ generated by $\coset{g}$.

%%%%%%%%%%%%%%%%%%%%%%%%%%%%%%%%
\begin{prop}\label{P-submoduleS}
Let~$\cM$ be a left submodule of~$\cS_a$.
\begin{arabiclist}
\item Then~$\cM=\lideal{\coset{g}}$, where~$g\in\cR$ is the unique monic polynomial of smallest degree such that
      $\coset{g}\in\cM$.
      Moreover,
      \begin{romanlist}
      \item $g\rmid f$ for any $f\in\cR$ such that $\coset{f}\in\cM$.
            In particular, $g\rmid (x^n-a)$.
      \item $g$ is the unique monic right divisor of~$x^n-a$ such that $\lideal{\coset{g}}=\cM$.
      \end{romanlist}
\item Let~$f\in\cR$. Then $\lideal{\coset{f}}=\lideal{\coset{g}}$, where $g=\gcrd(f,\,x^n-a)$.
\end{arabiclist}
\end{prop}
%%%%%%%%%%%%%%%%%%%%%%%%%%%%%%%%

We mention in passing that in the central case (see Remark~\ref{R-twosided}) the ring~$\cS_a$ is Frobenius.
This is a trivial consequence of the fact that~$\cS_a$ is finite and by Proposition~\ref{P-submoduleS}(1)
a principal left ideal ring; see \cite[Th.~1]{Hon01}.

Let us turn to the general case again.
The following is now immediate.
We use the notation $\im(M)$ for the rowspace of a matrix~$M$.
%%%%%%%%%%%%%%%%%%%%%%%%%%%%%%
\begin{cor}[see also \mbox{\cite{BoUl09}}]\label{C-Mbasis}
Let~$g\in\cR$ be a right divisor of $x^n-a$, and let $\deg(g)=r$. Set $\cM:=\lideal{\coset{g}}$.
Then~$\cM$ is a left~$\F$-vector space of dimension~$k:=n-r$ with basis
$\{\coset{g},\,\coset{xg},$ $\ldots,\,\coset{x^{k-1}g}\}$.
Writing $g=\sum_{i=0}^r g_i x^i$, we conclude
\[
  \vva(\cM)=\im(M),
\]
where
\begin{equation}\label{e-M}
  M=\begin{pmatrix}\vva(\coset{g})\\ \vva(\coset{xg})\\ \vdots\\ \vva(\coset{x^{k-1}g})\end{pmatrix}
  =\begin{pmatrix}g_0&g_1        &\cdots       &g_r     &        &   &\\
                    &\theta(g_0)&\theta(g_1)&\cdots&\theta(g_r)&   &\\
                    &           &\ddots     &\ddots   &     &\ddots&\\
                    &           &           &\!\!\theta^{k-1}(g_0)&\!\!\theta^{k-1}(g_1)&\cdots&\theta^{k-1}(g_r)
  \end{pmatrix}\in\F^{k\times n}.
\end{equation}
\end{cor}
%%%%%%%%%%%%%%%%%%%%%%%%%%%%%%
\begin{proof}
Let $hg=x^n-a$. Consider $f\coset{g}\in\cM$.
Division with remainder of~$f$ by~$h$ yields $f=th+s$ for some $t,s\in\cR$ with $\deg(s)<\deg(h)=k$.
Then $\coset{fg}=\coset{thg+sg}=\coset{sg}$, and the latter is in the span of
$\{\coset{g},\,\coset{xg},\ldots,\coset{x^{k-1}g}\}$.
Linear independence is clear from the matrix~$M$.
\end{proof}

We close this section with the definition of $(\theta,a)$-constacyclicity and an illustrating example.
The definition is a special case of \cite[Def.~1]{BoUl09}.
%%%%%%%%%%%%%%%%%%%%%%%%%%%%%%
\begin{defi}\label{D-CBC}
A subspace $\cC\subseteq\F^n$ is called $(\theta,a)$-\emph{constacyclic} if
$\ppa(\cC)$ is a submodule of~$\cS_a$.
The code~$\cC\subseteq\F^n$ is called \emph{skew-constacyclic} if it is $(\theta,a)$-constacyclic for some $\theta\in\Aut(\F)$
and $a\in\F^*$.
The code is called $\theta$-\emph{cyclic} if it is $(\theta,1)$-\emph{constacyclic}.
\end{defi}
%%%%%%%%%%%%%%%%%%%%%%%%%%%%%
It is easy to see~\cite[Sec.~2]{BoUl11} that a subspace~$\cC\subseteq\F^n$ is $(\theta,a)$-constacyclic if and only if
\begin{equation}\label{e-shift}
   (f_0,\ldots,f_{n-1})\in\cC\Longrightarrow (a\theta(f_{n-1}),\theta(f_0),\ldots,\theta(f_{n-2}))\in\cC.
\end{equation}

It is an immediate consequence of Proposition~\ref{P-submoduleS}(1) that if a subspace~$\cC\subseteq\F^n$, where $\{0\}\subsetneq\cC\subsetneq\F^n$,
is $(\theta,a)$-constacyclic and $(\theta,b)$-constacyclic, then $a=b$.
Furthermore, a $(\theta,a)$-constacyclic code has a generator matrix of the form~$M$ as in~\eqref{e-M}.
It is interesting to note that this matrix does not depend on~$a$.
The dependence on~$a$ materializes only through the fact that the code~$\im M$ is $(\theta,a)$-constacyclic, see~\eqref{e-shift}.
Indeed, let $\cC=\im M$, where $M\in\F^{k\times n}$ has a form as in~\eqref{e-M}, and without loss of generality assume $g_r=1$.
Let~$g:=\sum_{i=0}^r g_i x^i$.
The form of the matrix implies $r=n-k$.
Moreover, it shows that~$g$ is the unique monic polynomial of smallest degree in $\ppa(\cC)$.
As a consequence, Proposition~\ref{P-submoduleS}(1) implies that~$\cC$ is $(\theta,a)$-constacyclic if and only if
the polynomial~$g$ is a right divisor of $x^n-a$ of degree $n-k$.

Proposition~\ref{P-submoduleS} tells us that,
as in the classical commutative case, the $(\theta,a)$-constacyclic codes in~$\F^n$ are in bijection with the
distinct monic right divisors of~$x^n-a$.
However, as is well known, skew-polynomials do not factor uniquely into irreducible polynomials (but see also \cite[Thm.~1, Page~494]{Ore33}),
which often results in a large number of right divisors.
We provide the following small example, which will be used again in later sections.

%%%%%%%%%%%%%%%%%%%%%%%
\begin{exa}\label{E-FactorsXna}
Consider the field~$\F_8=\F_2[\alpha]$, where $\alpha^3=\alpha+1$, and let~$\theta$ be the Frobenius homomorphism on~$\F_8$,
thus $\theta(c)=c^2$ for all $c\in\F_{8}$.
Let $f:=x^7+\alpha$.
With the aid of an exhaustive search one finds that~$f$ has the monic right divisors
\begin{align*}
  &g^{(0)}=1,\quad g^{(1)}=x+\alpha,\quad g^{(2)}=x^3+\alpha^4 x^2+1,\quad g^{(3)}=x^3+\alpha^6x+1,\\[.6ex]
  &g^{(4)}=x^4+\alpha x^3+\alpha^5x^2+\alpha,\quad g^{(5)}=x^4+\alpha^5x^2+x+\alpha,\\[.6ex]
  &g^{(6)}=x^6+\alpha^4 x^5+\alpha^6 x^4+x^3+\alpha^4 x^2+\alpha^6x+1,\quad g^{(7)}=x^7+\alpha.
\end{align*}
The polynomials $g^{(2)},\,g^{(3)},\,g^{(6)}$ are not left divisors of $x^7+\alpha$, while all others are.
Moreover, we have the lattice shown in Figure~\ref{F-Lattices1} with respect to right division, which in turn provides
us with the lattice of the $(\theta,\alpha)$-constacyclic codes $\cC^{(i)}:=\vva\big(\lideal{\coset{g^{(i)}}}\big)$
in~$\F_8^7$ with respect to inclusion.
\begin{figure}[ht]
\centering
  \mbox{}\hspace*{1cm}
  \includegraphics[height=5cm]{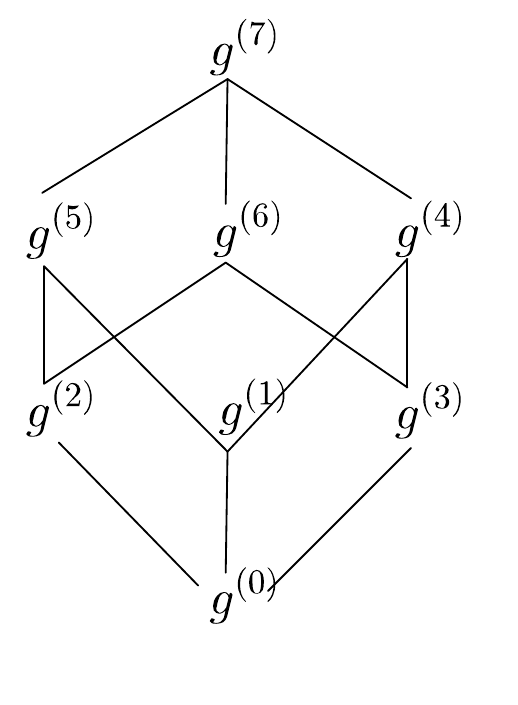}
         \hspace*{2cm}\includegraphics[height=5cm]{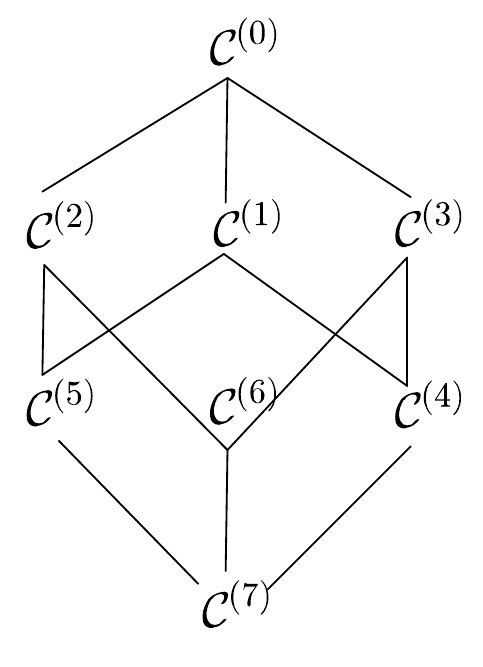}
    \caption{Lattice of monic right divisors of $x^7+\alpha$ and the corresponding codes}
    \label{F-Lattices1}
\end{figure}
\\
This means, for instance, that $g^{(1)}$ is a right divisor of~$g^{(5)}$ and thus $\cC^{(5)}\subseteq\cC^{(1)}$.
The latter implies that $(\cC^{(1)})^{\perp}\subseteq(\cC^{(5)})^{\perp}$.
The lattice of right divisors (in a suitable skew polynomial ring) corresponding to the dual codes
will be provided in Section~\ref{S-lattice}.

It is worth noting that the codes generated by $g^{(2)},\ldots,\,g^{(5)}$ are near-MDS (but not MDS), that is,
both the code and its dual have defect~$1$ (recall that the defect of a code is the difference between
the Singleton bound and the distance of the code).
The codes generated by~$g^{(1)}$ and~$g^{(6)}$ are trivial MDS codes.
\end{exa}
%%%%%%%%%%%%%%%%%%%%%%%%%%%%%%%%%%%%%%%%%%%%%%

Of course, as in the classical commutative case, general skew-constacyclic codes are not MDS or otherwise optimal.
In fact, as has been observed already by Boucher/Ul\-mer~\cite[Tables~1~--~3]{BoUl14},
for many choices of~$n$ there are no skew-constacyclic codes of length~$n$ that have the best
possible distance among all codes with the same parameters $(q,n,k)$.
But at the same time there are plenty of parameters for which skew-constacyclicity leads to the best codes known.
Tables can be found in~\cite{BGU07,CLU09}.

\section{Circulants}\label{S-Circ1}
In this section, we associate with each coset~$\coset{f}\in\cS_a=\cR/\lideal{x^n-a}$ a circulant.
This is a matrix in~$\F^{n\times n}$ whose rows reflect the module structure in~$\cS_a$ and
its row space is, up to the isomorphism~$\ppa$, the left submodule of~$\cS_a$ generated by~$\coset{f}$.
The situation becomes particularly nice when $x^n-a$ is central, in which case the circulant provides a
ring embedding of~$\cS_a$ as a subring in~$\F^{n\times n}$.

As before, let $\cR=\F[x;\theta]$ and $\cS_a=\cR/\lideal{x^n-a}$ for some fixed $a\in\F^*$.
Recall the left $\F$-isomorphism~$\ppa$ and its inverse~$\vva$ from~\eqref{e-ppa} and~\eqref{e-vva}.
These maps give rise to the following circulant matrices.
%%%%%%%%%%%%%%%%%%%%%%%%%%%%%%
\begin{defi}\label{D-CircMat}
For $\coset{f}\in\cS_a$ define the \emph{$(\theta,a)$-circulant}
\begin{equation}\label{e-Circat}
    M_a^{\theta}(\coset{f}):=\begin{pmatrix}\vv_a(\coset{f})\\ \vv_a(x\coset{f})\\ \vdots \\
                  \vv_a(x^{n-2}\coset{f})\\ \vv_a(x^{n-1}\coset{f})\end{pmatrix}\in\F^{n\times n}.
\end{equation}
Thus we have a map
\[
   M_a^{\theta}:\;\cS_a\longrightarrow \F^{n\times n},\quad \coset{f}\longmapsto M_a^{\theta}(\coset{f}).
\]
\end{defi}
%%%%%%%%%%%%%%%%%%%%%%%%%%%%%%
Explicitly, the circulant of~$\coset{f}$ is given as follows.
Without loss of generality assume $\deg(f)<n$ and thus $f=\sum_{i=0}^{n-1}f_ix^i$.
For any $i\in\N_0$ and $\gamma\in\F$ we have $\coset{x\gamma x^{i}}=\coset{\theta(\gamma)x^{i+1}}$  and hence
$\coset{x\gamma x^{n-1}}=\coset{\theta(\gamma)x^n}=\coset{\theta(\gamma)a}$.
This leads to
\begin{equation}\label{e-Circatc}
     M_a^{\theta}(\coset{f})=\begin{pmatrix}f_0\!&\!f_1\!&\!f_2\!&\!\ldots\!&\!f_{n-2}\!&\!f_{n-1}\\
           a\theta(f_{n-1})\!&\!\theta(f_0)\!&\!\theta(f_1)\!&\!\ldots \!&\!\theta(f_{n-3})\!&\!\theta(f_{n-2})\\
           a \theta^2(f_{n-2})\!&\!\theta(a)\theta^2(f_{n-1})\!&\!\theta^2(f_0)\!&\!\ldots \!&\!\theta^2(f_{n-4})\!&\!\theta^2(f_{n-3})\\
                          \vdots \!&\!\vdots \!&\!\ddots  \!& &\!\vdots\!&\!\vdots \\
           a\theta^{n-2}(f_2)\!&\!\theta(a)\theta^{n-2}(f_3)\!&\!\theta^2(a)\theta^{n-2}(f_4)\!&\!\ldots\!&\!\theta^{n-2}(f_0)
                                \!&\!\theta^{n-2}(f_1)\\
          a\theta^{n-1}(f_1)\!&\!\theta(a)\theta^{n-1}(f_2)\!&\!\theta^2(a)\theta^{n-1}(f_3)\!&\!\ldots\!&\!
                                 \theta^{n-2}(a)\theta^{n-1}(f_{n-1})\!&\!\theta^{n-1}(f_0)
      \end{pmatrix}.
\end{equation}
In other words, $\CircMatat(\coset{f})=(M_{ij})_{i,j=0,\ldots,n-1}$, where
\begin{equation}\label{e-Mij}
    M_{ij}=\left\{\begin{array}{ll} \theta^i(f_{j-i}),&\text{if }i\leq j,\\[.7ex]
                               \theta^j(a)\theta^i(f_{n+j-i}),&\text{if }i>j.\end{array}\right.
\end{equation}
For example,
\[
  \CircMatat(\coset{x})=\begin{pmatrix}   &1& & &\\ & &1& & \\ & & & \ddots & \\ & & & &1\\ a& & & &\end{pmatrix}\ \text{ and }\
  \CircMatat(\coset{x^2})=\begin{pmatrix}  & &\!\!1& & \\ & & &\ddots &\\ & & & &1\\ a& & & &\\ &\!\!\theta(a)\!\!& & &\end{pmatrix}.
\]

%%%%%%%%%%%%%%%%%%%%%%%%
\begin{rem}\label{R-PropCircMat}
\begin{alphalist}
\item The map $\CircMatat$ is injective and additive, i.e.,
      \[
        \CircMatat(\coset{f}+\coset{f'})=\CircMatat(\coset{f})+\CircMatat(\coset{f'})\text{ for all }f,\,f'\in\cR.
      \]
\item $\CircMatat(c\coset{f})=\CircMat{b}(\coset{c})\CircMatat(\coset{f})$ for all $c\in\F$ and $f\in\cR$ and all~$b\in\F^*$.
      This follows directly from the definition along with the fact that
      \begin{equation}\label{e-CircMatConst}
          \CircMat{b}(\coset{c})=\begin{pmatrix}c& & & \\ & \theta(c)& & \\ & &\ddots& \\ & & &\theta^{n-1}(c)\end{pmatrix}
          \text{ for any }b\in\F^*.
      \end{equation}
      As a consequence, $\CircMatat$ is not $\F$-linear (unless $\theta=\text{id}_{\F}$), but it is
     $\Fix_{\F}(\theta)$-linear.
\item The map~$\CircMatat$ is not multiplicative, that is,
      $\CircMatat(\coset{ff'})\neq\CircMatat(\coset{f})\CircMatat(\coset{f'})$ in general.
      This simply reflects the fact that~$\cS_a$ is not a ring.
\end{alphalist}
\end{rem}
%%%%%%%%%%%%%%%%%%%%%%%%%%%%%%%%%%%

As a particular case of Part~(c) above, we observe that the identity $hg=x^n-a$ does not imply
$\CircMatat(\coset{h})\CircMatat(\coset{g})=0$.
(For an example take the right divisor $g=x+\alpha^5$ of $x^5-\alpha\in\F_8[x;\theta]$,
where~$\theta$ is the Frobenius homomorphism and $\alpha^3+\alpha+1=0$.)
The situation becomes much nicer when $x^n-a$ is central, as we will see in Theorem~\ref{T-centralCirc}.
For the general case we will establish a certain product formula later in Theorem~\ref{T-ProdCirc}.

The next result shows that the row space of the circulant corresponds to the left submodule under the isomorphism~$\vva$.
%%%%%%%%%%%%%%%%%%%%%%%%%%%%%%
\begin{prop}\label{P-basicCM}
We have
\[
   \ppa(u\CircMatat(\coset{f}))=\ppa(u)\coset{f} \text{ for all }u\in\F^n\text{ and }f\in\cR.
\]
As a consequence, $\im \CircMatat(\coset{f})=\vva(\lideal{\coset{f}})$.
\end{prop}
%%%%%%%%%%%%%%%%%%%%%%%%%%%%%
\begin{proof}
Writing $u=(u_0,\ldots,u_{n-1})$ we compute
$u\CircMatat(\coset{f}))=\sum_{i=0}^{n-1}u_i\vva(x^i\coset{f})=\vva(\ppa(u)\coset{f})$.
This proves the first statement.
The containment ``$\subseteq$'' of the second statement is an immediate consequence.
As for ``$\supseteq$'' consider $\coset{hf}\in\lideal{\coset{f}}$ for some $h\in\cR$.
If we can show that $\coset{hf}=\coset{kf}$ for some $k\in\cR$ with $\deg(k)<n$,
then the first part yields $\vva(\coset{hf})=\vva(\coset{kf})=\vva(\coset{k})\CircMatat(\coset{f})$, as desired.
For the existence of such~$k$, let $uf=v(x^n-a)=\lclm(f,\,x^n-a)$ with some $u,v\in\cR$ and where $\deg(u)\leq n$.
Such polynomials exist due to Remark~\ref{R-PropR}(c).
Using right division with reminder we obtain $h=qu+k$ for some $q,\,k\in\cR$ with $\deg(k)<n$.
Then $\coset{hf}=\coset{quf+kf}=\coset{qv(x^n-a)+kf}=\coset{kf}$, as desired.
\end{proof}

The last proposition and Proposition~\ref{P-submoduleS}(2) provide us with the following.
%%%%%%%%%%%%%%%%%%%%%%%%%%%%%%
\begin{cor}\label{C-CircSpaces}
\begin{alphalist}
\item Let $f,\,g\in\cR$. Then $\im \CircMatat(\coset{f})\subseteq\im \CircMatat(\coset{g})\Longleftrightarrow g\rmid f$.
\item Let $f\in\cR$ and $g=\gcrd(f,\,x^n-a)$. Then $\im \CircMatat(\coset{f})=\im \CircMatat(\coset{g})$.
\end{alphalist}
\end{cor}
%%%%%%%%%%%%%%%%%%%%%%%%%%%%%%

Note that $\im \CircMatat(\coset{f})\subseteq\im \CircMatat(\coset{g})$ if and only if $\CircMatat(\coset{f})=Q \CircMatat(\coset{g})$
for some $Q\in\F^{n\times n}$.
Therefore,~(a) above may be rephrased as
\begin{equation}\label{e-divisors}
  g\rmid f\Longleftrightarrow \CircMatat(\coset{g})\rmid \CircMatat(\coset{f}),
\end{equation}
that is,~$g$ is a right divisor of~$f$ in the ring~$\cR$ if and only if $\CircMatat(\coset{g})$ is a right divisor of
$\CircMatat(\coset{f})$ in the ring $\F^{n\times n}$.
In other words, $\CircMatat$ induces an isomorphism between the lattice of monic polynomials in~$\cR$ with right division
and the lattice of associated circulants in $\F^{n\times n}$ with right division.
In Theorem~\ref{T-ProdCirc} we will see that if~$g$ is a right divisor of~$x^n-a$ then the matrix~$Q$ above may be chosen
as a particular circulant as well.
We will also see that if~$g$ is not a right divisor of~$x^n-a$ then the matrix~$Q$ cannot be chosen as a
circulant matrix in general.

Combining Corollary~\ref{C-Mbasis}, Propositions~\ref{P-submoduleS},~\ref{P-basicCM}, and Corollary~\ref{C-CircSpaces} we obtain the following
description of $(\theta,a)$-constacyclic codes.

%%%%%%%%%%%%%%%%%%%%%%%%%%%%%%%%
\begin{theo}\label{T-tacyclic}
Let~$g\in\cR$ be a right divisor of~$x^n-a$ of degree~$n-k$.
Then the circulant $\CircMatat(\coset{g})$ has rank~$k$ and its first~$k$ rows form a basis of
the $(\theta,a)$-constacyclic code $\vva(\lideal{\coset{g}})$.
As a consequence, the $(\theta,a)$-constacyclic codes in~$\F^n$ are exactly the subspaces
$\im \CircMatat(\coset{g})$, where~$g$ is a monic right divisor of~$x^n-a$.
Different such divisors result in different codes.
We call~$g$ the generator polynomial of the code~$\im \CircMatat(\coset{g})$.
\end{theo}
%%%%%%%%%%%%%%%%%%%%%%%%%%%%%%%

In the case where $x^n-a$ is central (see Remark~\ref{R-twosided}) we obtain a particularly nice situation for the circulants.
%%%%%%%%%%%%%%%%%%%%%%%%%%%%%%
\begin{theo}\label{T-centralCirc}
Let $x^n-a$ be central; thus $\cS_a$ is a ring. Then
\[
   \CircMatat(\coset{fg})=\CircMatat(\coset{f})\CircMatat(\coset{g})\text{ for all }f,\,g\in\cR.
\]
Hence~$\CircMatat$ is a ring isomorphism between~$\cS_a$ and the subring $\CircMatat(\cS_a)\subseteq\F^{n\times n}$.
\end{theo}
%%%%%%%%%%%%%%%%%%%%%%%%%%%%%%
\begin{proof}
With the aid of Proposition~\ref{P-basicCM} we compute
$\ppa\big(u\CircMatat(\coset{f})\CircMatat(\coset{g})\big)=\ppa\big(u\CircMatat(\coset{f})\big)\coset{g}
  =\ppa(u)\coset{f}\coset{g}=\ppa(u)\coset{fg}=\ppa(u\CircMatat(\coset{fg}))$
for all $u\in\F^n$. This shows the desired result.
\end{proof}

In order to derive further results on circulants, we need some identities pertaining to factorizations of $x^n-a$.
They will be collected in the next section, and we return to circulants thereafter.

%%%%%%%%%%%%%%%%%%%%%%%%%%%%%%%%%%%%%%%%%%%%%%%%%%%%%%
\section{Factorizations of $x^n-a$}\label{S-xna}
Again, we consider the skew-polynomial ring $\cR:=\F[x;\theta]$ for some fixed $\theta\in\Aut(\F)$.
In this section we study factorizations of the form $x^n-a=hg$ in~$\cR$.
They give rise to an abundance of further factorizations and lead to various identities for the coefficients of~$h$ and~$g$.
In order to derive these results we need the following maps.

The natural extension of~$\theta$ to~$\cR$ will be denoted by~$\theta$ as well, thus
\begin{equation}\label{e-thetaR}
   \theta:\;\cR\longrightarrow\cR,\quad \sum_{i=0}^r f_ix^i\longmapsto \sum_{i=0}^r\theta(f_i)x^i.
\end{equation}
As a consequence,
\begin{equation}\label{e-thetaRf}
    xf=\theta(f)x\text{ for all }f\in\cR.
\end{equation}
In addition, on the ring of skew-Laurent polynomials $\F[x,x^{-1};\theta]$ we consider the map
\begin{equation}\label{e-phiL}
    \varphi:\;\F[x,x^{-1};\theta]\longrightarrow\F[x,x^{-1};\theta],\quad
      \sum_{i=m}^n a_i x^i\longmapsto \sum_{i=m}^n x^{-i}a_i.
\end{equation}
It gives rise to two reciprocal polynomials, a \emph{left reciprocal}~$\rho_l$ and a \emph{right reciprocal}~$\rho_r$, defined as follows:
\begin{equation}\label{e-rholr}
  \rho_l:\;\cR\longrightarrow\cR,\quad f\longmapsto x^{\deg f}\varphi(f)\ \text{ and }\
  \rho_r:\;\cR\longrightarrow\cR,\quad f\longmapsto \varphi(f)x^{\deg f}.
\end{equation}
Explicitly these maps are given by
\begin{equation}\label{e-rholrf}
  \rho_l\Big(\sum_{i=0}^t f_ix^i\Big)= \sum_{i=0}^t x^{t-i}f_i=\sum_{i=0}^t\theta^i(f_{t-i}) x^{i}\ \text{ and }\
  \rho_r\Big(\sum_{i=0}^t f_ix^i\Big)= \sum_{i=0}^t\theta^{i-t}(f_{t-i})x^i
\end{equation}
where $f_t\neq0$.
The left reciprocal and its multiplicativity rule in~(h) of the following proposition appear also in
\cite[Def.~3, Lem.~1]{BoUl11}.

%%%%%%%%%%%%%%%%%%%%%%%%%%%%%%%%%%%%%%%%%%%%%%%%
\begin{prop}\label{P-basics}\
\begin{alphalist}
\item $\theta$ is a ring isomorphism of~$\cR$.
\item $\varphi$ is a ring anti-isomorphism:
      $\varphi(f+f')=\varphi(f)+\varphi(f')$ and $\varphi(ff')=\varphi(f')\varphi(f)$
      for all $f,\,f'\in\F[x,x^{-1};\theta]$.
\item $\rho_r|_{\F}={\rm id}_{\F}=\rho_l|_{\F}$.
\item $\rho_l(f)=\theta^{\deg(f)}(\rho_r(f))$ for all $f\in\cR$.
\item $\theta\circ\rho_l=\rho_l\circ\theta$ and $\theta\circ\rho_r=\rho_r\circ\theta$.
\item $\rho_l\circ\rho_l(f)=\theta^{\deg f}(f)$ and $\rho_r\circ\rho_r(f)=\theta^{-\deg f}(f)$ for all $f\in\cR$.
\item $\rho_r\circ\rho_l=\rho_l\circ\rho_r={\rm id}_{\cR}$.
\item $\rho_l(f_1f_2)=\theta^{k_1}(\rho_l(f_2))\rho_l(f_1)$ and $\rho_r(f_1f_2)=\rho_r(f_2)\theta^{-k_2}(\rho_r(f_1))$
      for all $f_1,\,f_2\in\cR$ and where $k_i=\deg f_i$.
\end{alphalist}
\end{prop}
%%%%%%%%%%%%%%%%%%%%%%%%%%%%%%%%%%%%%%%%
\begin{proof}
(a) and~(c) are obvious.
The additivity in~(b) is clear, and for the multiplicativity is suffices to show that
$\varphi(ax^mbx^n)=\varphi(bx^n)\varphi(ax^m)$, which can easily be verified.
(d) and~(e) are immediate from~\eqref{e-rholrf}.
For~(f) we compute $\rho_l(\rho_l(\sum_{i=0}^tf_ix^i))=\rho_l(\sum_{i=0}^t\theta^i(f_{t-i})x^i)=
 \sum_{i=0}^t\theta^i(\theta^{t-i}(f_i))x^i=\theta^t(f)$.
Similarly we have $\rho_r(\rho_r(\sum_{i=0}^tf_ix^i))=\rho_r(\sum_{i=0}^t\theta^{i-t}(f_{t-i})x^i)$
 $=\sum_{i=0}^t\theta^{i-t}(\theta^{-i}(f_i))x^i=\theta^{-t}(f)$.
(g) follows from~(d),~(e), and~(f).
For~(h) we use~\eqref{e-thetaRf} and the previous properties to compute
$\rho_l(f_1f_2)=x^{k_1+k_2}\varphi(f_1f_2)=x^{k_1+k_2}\varphi(f_2)\varphi(f_1)
 = x^{k_2}\theta^{k_1}(\varphi(f_2))x^{k_1}\varphi(f_1)=\rho_l(\theta^{k_1}(f_2))\rho_l(f_1)$.
The second identity follows from the first one using~(d).
\end{proof}

Now we turn to an identity of the form $x^n-a=hg$ and derive various consequences.
We introduce the notation
\begin{equation}\label{e-gamma}
   \gamma(a,g):=ag_0^{-1}\theta^n(g_0)\text{ for any right divisor~$g$ of $x^n-a$},
\end{equation}
where~$g_0$ is the constant coefficient of~$g$.
One may note that $\gamma(a,g)$ is the conjugate $a^{g_0}$ in the skew-polynomial ring $\F[x;\theta^n]$ in the sense
of~\cite[Eq.~(2.5)]{LaLe88}.

%%%%%%%%%%%%%%%%%%%%%%%%%%%%
\begin{theo}[\mbox{see also \cite[Lem.~2]{BoUl11}}] \label{T-factors}
Let $a\in\F^*$ and $g=\sum_{i=0}^{n-k}g_ix^i,\,h=\sum_{i=0}^k h_i x^i\in\cR$ such that $\deg(h)=k$ and $\deg(g)=n-k$.
Define $c=\gamma(a,g)$.
Then the following are equivalent.
\begin{arabiclist}
\item $x^n-a=hg$,
\item $x^n-c=\theta^n(g)h$,
\item $x^n-\theta^{-n}(c)=g\theta^{-n}(h)$,
\end{arabiclist}
Furthermore, if any, hence all, of the above is true then
\begin{equation}\label{e-thetahg}
  \theta^n(g)a=cg\ \text{ and }\ a\theta^{-n}(h)=h\theta^{-n}(c).
\end{equation}
\end{theo}
%%%%%%%%%%%%%%%%%%%%%%%%%%%%%

\begin{proof}
(1)~$\Rightarrow$~(2)
Left-multiplying $x^n-a=hg$ with $\theta^n(g)$ and using $\theta^n(g)x^n=x^ng$, we obtain
$(x^n-\theta^n(g)h)g=\theta^n(g)a$.
This shows that~$g$ is a right divisor of $\theta^n(g)a$.
Since both polynomials have the same degree we conclude $cg=\theta^n(g)a$
with~$c$ as in the theorem.
Now we have $(x^n-\theta^n(g)h)g=cg$, and cancellation of~$g$ results in $x^n-c=\theta^n(g)h$, as desired.
\\
(2)~$\Rightarrow$~(3) follows by applying~$\theta^{-n}$.
\\
(3)~$\Rightarrow$~(1) follows from using the implication~(1)~$\Rightarrow$~(2) along with $g_0h_0=-a$.
\\
It remains to show the identities in~\eqref{e-thetahg}. The first one has been derived already in the first part
of this proof.
For the second one we right-multiply~(1) by $\theta^{-n}(h)$ and compute
$a\theta^{-n}(h)=x^n\theta^{-n}(h)-hg\theta^{-n}(h)=h(x^n-g\theta^{-n}(h))=h\theta^{-n}(c)$, where the last step follows from~(3).
\end{proof}

At the end of this section we will elaborate on how the search for all right factors of $x^n-a$ (thus of all $(\theta,a)$-constacyclic
codes) can be aided by the above theorem.

Comparing left coefficients in the identities in~\eqref{e-thetahg} yields
%%%%%%%%%%%%%%%%%%%%%%%%%
\begin{cor}\label{C-Coeffg}
Let~$a\in\F^*$ and $g,\,h\in\cR$ such that $x^n-a=hg$ and let $c=\gamma(a,g)$.
Write $g=\sum_{i=0}^{n-k}g_ix^i$ and $h=\sum_{i=0}^k h_ix^i$. Then
\[
  cg_t=\theta^t(a)\theta^n(g_t)\text{ and }  a\theta^{-n}(h_t)=h_t\theta^{t-n}(c)
  \text{ for all }t\geq0.
\]
\end{cor}
%%%%%%%%%%%%%%%%%%%%%%%%%

The following additional identities will be crucial in the next sections when turning to transpositions of circulants and duals of
$\theta$-constacyclic codes.
%%%%%%%%%%%%%%%%%%%%%%%%%%%
\begin{cor}\label{C-hhataghat}
Let~$a\in\F^*$ and $g,\,h\in\cR$ such that $x^n-a=hg$ and let $c=\gamma(a,g)$.
Define
\[
     \widehat{h}^{^l}:=\rho_l(\theta^{-n}(h))\ \text{ and }\ \widehat{g}^{r}:=\rho_r(\theta^n(g)).
\]
Then
\begin{alphalist}
\item $ga^{-1}h=c^{-1}(x^n-c)$,
\item $-\theta^{k-n}(c^{-1})\theta^k(c^{-1})\widehat{h}^{^l}a\widehat{g}^{r}=x^n-\theta^k(c^{-1})$,
\item $-\widehat{g}^{r}\theta^{k-n}(c^{-1})\widehat{h}^{^l}=x^n-a^{-1}$.
\end{alphalist}
\end{cor}
%%%%%%%%%%%%%%%%%%%%%%%%%%%%%%
\begin{proof}
(a) Using~\eqref{e-thetahg} and~(2) of Theorem~\ref{T-factors} we compute $ga^{-1}h=c^{-1}\theta^n(g)h=c^{-1}(x^n-c)$.
\\
(b)
Applying~$\rho_l$ to~(a) yields $-x^n+c^{-1}=\rho_l(ga^{-1}h)=\theta^{n-k}(\rho_l(a^{-1}h))\rho_l(g)$ by virtue of
Proposition~\ref{P-basics}(h).
Applying~$\theta^k$ and using that $\theta^k(\rho_l(g))=\rho_r(\theta^n(g))=\widehat{g}^{r}$, we obtain
$-\theta^n(\rho_l(a^{-1}h))\widehat{g}^{r}=x^n-\theta^k(c^{-1})$.
Hence it remains to show that
$\theta^{k-n}(c)\theta^{k}(c)\theta^n(\rho_l(a^{-1}h))=\widehat{h}^{^l}a$.
First observe that $\theta^n(\rho_l(a^{-1}h))=\rho_l(\theta^n(a^{-1}h))=\rho_l(hc^{-1})$ due to~\eqref{e-thetahg}.
Using again Proposition~\ref{P-basics}(h) and once more~\eqref{e-thetahg} we derive
$\theta^{k-n}(c)\theta^k(c)\rho_l(hc^{-1})=\theta^{k-n}(c)\rho_l(h)=\rho_l(h\theta^{-n}(c))
=\rho_l(a\theta^{-n}(h))=\widehat{h}^{^l}a$, and this establishes~(b).
\\
(c) We apply $\rho_l$ to Theorem~\ref{T-factors}(1) to obtain $-x^na+1=\rho_l(hg)=\rho_l(\theta^k(g))\rho_l(h)$.
Thus, $x^n-a^{-1}=-\widehat{g}^{r}\rho_l(h)a^{-1}=-\widehat{g}^{r}\rho_l(a^{-1}h)$.
By~\eqref{e-thetahg} we have $a^{-1}h=\theta^{-n}(h)\theta^{-n}(c^{-1})$, and thus
$x^n-a^{-1}=-\widehat{g}^{r}\rho_l(\theta^{-n}(h)\theta^{-n}(c^{-1}))=-\widehat{g}^{r}\theta^{k-n}(c^{-1})\widehat{h}^{^l}$, as desired.
\end{proof}

%%%%%%%%%%%%%%%%%%%%%%%%%%
\begin{rem}\label{R-constant}
Let~$a\in\F^*$ and $g,\,h\in\cR$ such that $x^n-a=hg$ and let $c=\gamma(a,g)$. Suppose $g$ and~$h$ are monic.
Then $c=\theta^{n-k}(a)$, which follows from $t=n-k$ in Corollary~\ref{C-Coeffg}.
As a consequence, the constant $\theta^k(c^{-1})$ in Corollary~\ref{C-hhataghat}(b) equals $\theta^n(a^{-1})$ and is thus independent
of the choice of~$g,\,h$ and the degree~$k$.
\end{rem}
%%%%%%%%%%%%%%%%%%%%%%%%%%

The rest of this section is devoted to a brief discussion of how to find all right divisors of the polynomials of the form $x^n-a$.
For the general factorization problem in~$\F[x;\theta]$ and fast algorithms we refer to~\cite{Gie98,CaBo12}.

A major cost saver for finding all right divisors is obtained from Theorem~\ref{T-factors}.
Indeed, note that if~$g_0=1$ then $c=a$ and the implication~(1)~$\Rightarrow$~(2) of that theorem shows that the left
divisor~$h$ of~$x^n-a$ is also a right divisor.
Thus, in order to determine all right divisors of~$x^n-a$ it suffices to compute all right divisors,~$g$, up to degree
$\lfloor n/2\rfloor$ with constant term~$1$; the corresponding left factors,~$h$, will then be the remaining right divisors
with degree at least $\lfloor n/2\rfloor$ (but in general not with constant term~$1$).

Next, we observe that
$x^n-a=hg\Longleftrightarrow x^n-ab\theta^n(b^{-1})=\big(\theta^n(b^{-1})h\big)gb$ for any $a,\,b\in\F^*$.
This is seen by right-multiplying $x^n-a=hg$ by~$b$ and left-multiplying by~$\theta^n(b^{-1})$.
Thus, the map $g\longmapsto gb$ provides us with a bijection between the right divisors of~$x^n-a$ and those of
$x^n-\hat{a}$, where $\hat{a}=ab\theta^n(b^{-1})$.
Note that the map
\begin{equation}\label{e-vartheta}
    \vartheta:\F^*\longrightarrow\F^*,\ b\longmapsto b\theta^n(b^{-1})
\end{equation}
is a group homomorphism with kernel
$\widetilde{\F}^*$, where $\widetilde{\F}=\Fix_{\F}(\theta^n)$.
As a consequence, by varying~$b$ we obtain for~$\hat{a}$ all values in the coset $a(\im\vartheta)$ in~$\F^*$.
This coset is exactly the set of all conjugates of~$a$ in~$\F[x;\theta^n]$ in the sense of \cite{LaLe88}.
All of this shows that factorizations of $x^n-a$ provide us easily with factorizations of $|\im\vartheta|$ distinct polynomials
of the form $x^n-\hat{a}$.

We summarize as follows.
%%%%%%%%%%%%%%%%%%%%%%%
\begin{prop}\label{P-RightDiv}
Let $a,\,b\in\F^*$ and set $\hat{a}:=ab\theta^n(b^{-1})$. Let~$g\in\cR$. Then
\[
   g\rmid (x^n-a)\Longleftrightarrow (gb)\rmid (x^n-\hat{a}).
\]
\end{prop}
%%%%%%%%%%%%%%%%%%%%%%%%%%%%%%
We will come back to this result in Theorem~\ref{T-MonomEquiv}, where we also relate the corresponding skew-constacyclic codes.

In addition to this result, Corollary~\ref{C-hhataghat} may provide additional information about the right divisors because it relates those of
$x^n-a$ to those of $x^n-a^{-1}$.
We illustrate all of this by some examples.

%%%%%%%%%%%%%%%%%%%%%%%%%%%%
\begin{exa}\label{E-RightDiv}
\begin{arabiclist}
\item Let $\text{char}(\F)=2$ and $\Fix_{\F}(\theta^n)=\F_2$. Then the map~$\vartheta$ is surjective and thus the set of right divisors of
      any~$x^n-a$ leads immediately to the set of all right divisors of $x^n-\hat{a}$ for any~$\hat{a}\in\F^*$.
      This is for instance the case for any field $\F_{2^p}$, where~$p$ is prime, along with any non-trivial automorphism~$\theta$
      and any~$n$ such that $p\nmid n$.
\item Let~$\F=\F_{16}$ and $\theta$ be the Frobenius map.  Let $n=6$.
      Then $\Fix_{\F}(\theta^6)=\F_4$.
      Thus~$\im(\vartheta)$ is the unique subgroup of~$\F^*$ of order $|\F^*|/|\F_4^*|=5$.
      Precisely, with~$\alpha$ being a primitive element of~$\F$ we have $\im\vartheta=\{1,\alpha^3,\alpha^6,\alpha^9,\alpha^{12}\}$, and
      the other two cosets are $\{\alpha,\alpha^4,\alpha^7,\alpha^{10},\alpha^{13}\}$ and $\{\alpha^2,\alpha^5,\alpha^8,\alpha^{11},\alpha^{14}\}$.
      One finds that $x^6-1$ has~$35$ distinct monic right divisors, and hence the same is true for $x^6-\alpha^{3i}$ for $i=1,\ldots,4$.
      One also finds that the polynomial $x^6-\alpha$ has no non-trivial right divisors.
      Now we may also use Corollary~\ref{C-hhataghat} and conclude that also $x^6-\alpha^{-1}$ has no non-trivial right divisors.
      Since $\alpha^{-1}=\alpha^{14}$, we conclude that $x^6-a$, where~$a$ is any element in the last two cosets has no non-trivial right
      divisors.
\item Let $\F=\F_9$ and $\theta$ be the Frobenius map.
      Let~$n=4$. Then $\theta^4=\text{id}$ and thus $\im(\vartheta)=\{1\}$.
      An exhaustive search shows that~$x^4-1$ has~$12$ monic right divisors, whereas $x^4-2$ has~$36$ such divisors.
\end{arabiclist}
\end{exa}
%%%%%%%%%%%%%%%%%%%%%%%%%%%%

%%%%%%%%%%%%%%%%%%%%%%%%%%%%%%%%%%%%%%%%%%%%%%%
\section{Circulants of right divisors of $x^n-a$}\label{S-Circ2}
As before, we consider the skew-polynomial ring $\cR:=\F[x;\theta]$ for some fixed $\theta\in\Aut(\F)$.
Recall from the paragraph right after Remark~\ref{R-PropCircMat} that in general $x^n-a=hg$ does not imply
$\CircMatat(\coset{h})\CircMatat(\coset{g})=0$.
In this section we will prove instead a specific product formula for circulants of right divisors of~$x^n-a$ that will be
sufficient for our investigation of skew-constacyclic codes.
Moreover, we will show that the transpose of such circulants is a circulant again.

Throughout, let $a\in\F^*$.
In order to compute modulo the left ideal $\lideal{x^n-a}$ we will need the following lemma.
%%%%%%%%%%%%%%%%%%%%%%%%%%%%
\begin{lemma}\label{L-xpowers}
In the left~$\cR$-module $\cS_a=\cR/\lideal{x^n-a}$ we have
\[
   \coset{x^{tn+j}}=\Big(\prod_{l=0}^{t-1}\theta^{ln+j}(a)\Big)\coset{x^j}\ \text{ for all }t\in\N,\,j=0,\ldots,n-1.
\]
\end{lemma}
%%%%%%%%%%%%%%%%%%%%%%%%%%%%
\begin{proof}
For $t=1$ we compute
$\coset{x^{n+j}}=\coset{x^j(x^n-a+a)}=\coset{x^ja}=\coset{\theta^j(a)x^j}$, as desired.
The rest follows similarly using induction on~$t$.
\end{proof}

We now turn to circulants of left multiples of~$\coset{g}$, where~$g$ is a right divisor of $x^n-a$.
Before presenting the general result, let us first compute the circulant of
$x\coset{g}$ in terms of the circulant of~$\coset{g}$.
%%%%%%%%%%%%%%%%
\begin{exa}\label{E-Circxg}
Let $x^n-a=hg$. Then $x^n=\theta^n(g)h+c$ by Theorem~\ref{T-factors}(2) and where~$c=\gamma(a,g)$; see~\eqref{e-gamma}.
This yields
\[
    \coset{x^ng}=\coset{\theta^n(g)hg+cg}=\coset{cg}\text{ in }\cS_a,
\]
and therefore
\[
   \CircMatat(\coset{xg})
    =\begin{pmatrix}\vv_a(x\coset{g})\\ \vv_a(x^2\coset{g})\\  \vdots \\  \vv_a(x^{n}\coset{g})\end{pmatrix}
     =\begin{pmatrix}   &1& & \\ & &\ddots & \\ & & &1\\ c& & & \end{pmatrix}
        \begin{pmatrix}\vv_a(\coset{g})\\ \vv_a(x\coset{g})\\  \vdots \\  \vv_a(x^{n-1}\coset{g})\end{pmatrix}
     =\begin{pmatrix}   &1& & \\ & &\ddots & \\ & & &1\\ c& & & \end{pmatrix}
        \CircMatat(\coset{g}).
\]
Note that this can be written as $\CircMatat(\coset{xg})=\CircMat{c}(\coset{x})\CircMatat(\coset{g})$, where
$\coset{x}:=x+\lideal{x^n-c}\in\cS_c$, while, as before, $\coset{g}=g+\lideal{x^n-a}\in\cS_a$.
\end{exa}
%%%%%%%%%%%%%%%%%
The product formula for circulants in the previous example can be generalized.
From now on we have to consider circulants for different bases and therefore use the convention that for a circulant
$\CircMat{b}(\coset{f})$ the coset~$\coset{f}$ is taken in~$\cS_b$, thus $\coset{f}=f+\lideal{x^n-b}$.
Recall the notation $\gamma(a,g)$ from~\eqref{e-gamma}.

%%%%%%%%%%%%%%%%%%%%%%%%%%%%%%%%
\begin{theo}\label{T-ProdCirc}
Let $x^n-a=hg$ and $f\in\cR$. Then
\[
   \CircMatat(\coset{fg})=\CircMat{c}(\coset{f})\CircMatat(\coset{g}), \text{ where }c=\gamma(a,g).
\]
\end{theo}
%%%%%%%%%%%%%%%%%%%%%%%%%%%%%%%
Note that if $|\theta|$ divides~$n$, then $c=a$ and thus
$\CircMatat(\coset{fg})=\CircMatat(\coset{f})\CircMatat(\coset{g})$ for all $f\in\cR$.
For the case where $x^n-a$ is central we have proven the same formula already for general~$g$ in
Theorem~\ref{T-centralCirc}.

\begin{proof}
Due to Remark~\ref{R-PropCircMat} it suffices to show the statement for $f=x^i$ for any $i\in\N_0$.
Write $i=tn+j$, where $0\leq j<n$.
Then Lemma~\ref{L-xpowers} yields $\coset{x^i}=d\coset{x^j}$, where $d:=\prod_{l=0}^{t-1}\theta^{ln+j}(c)$.
Thus, again Remark~\ref{R-PropCircMat}(b) shows that we may restrict ourselves to the case $0\leq i<n$.
Now we compute
\[
  \CircMat{c}(\coset{x^i})\CircMatat(\coset{g})
  =\!\begin{pmatrix} & & & &1& & & \\ & & & & &1 & & \\ & & & & & &\ddots & \\ & & & & & & &1\\
       \!c\!& & & & & & & \\  &\!\theta(c)\!& & & & & & \\  & &\!\ddots\!& & & & & \\   & & &\!\theta^{i-1}(c)\!& & & &
  \end{pmatrix}\!
  \begin{pmatrix}\vva(\coset{g})\\\vva(\coset{xg})\\ \vdots\\\!\vva(\coset{x^{i-1}g})\!\\\vva(\coset{x^{i}g})
          \\\!\vva(\coset{x^{i+1}g})\!\\ \vdots\\ \vva(\coset{x^{n-1}g})\end{pmatrix}
  =
  \begin{pmatrix}\vva(\coset{x^ig})\\ \vva(\coset{x^{i+1}g})\\ \vdots\\ \vva(\coset{x^{n-1}g})\\
          \vva(\coset{cg})\\  \vva(\coset{\theta(c)xg})\\ \vdots\\  \vva(\coset{\theta^{i-1}(c)x^{i-1}g})\end{pmatrix}.
\]
Using $\theta^l(c)x^l=x^lc$ as well as $c=x^n-\theta^n(g)h$ from Theorem~\ref{T-factors}(2), the
cosets modulo the left ideal $\lideal{x^n-a}$ satisfy
$\coset{\theta^{l}(c)x^{l}g}=\coset{x^lcg}=\coset{x^l(x^ng-\theta^n(g)hg)}=\coset{x^{n+l}g}$ for $l=0,\ldots,i-1$.
Hence the last matrix is $\CircMatat(\coset{x^ig})$, which is what we wanted.
\end{proof}

The leftmost matrix in above identity will be needed again. Clearly this matrix is invertible, and one easily verifies that
\begin{equation}\label{e-Mxi}
  \Big(\CircMat{b}(\coset{x^i})\T\Big)^{-1}=\CircMat{b^{-1}}(\coset{x^i}) \text{ for all }i=0,\ldots,n-1\text{ and any }b\in\F^*.
\end{equation}

Before we move on to discuss the transpose of a circulant, we take a brief digression and consider the situation of
Proposition~\ref{P-RightDiv} again.
%%%%%%%%%%%%%%%%%%%%%%%%%%%%%
\begin{theo}\label{T-MonomEquiv}
Let $x^n-a=hg$ and $b\in\F^*$. Then $gb\rmid (x^n-\hat{a})$, where $\hat{a}=\gamma(a,b^{-1})=ab\theta^n(b^{-1})$, and
\[
   \CircMat{\hat{a}}(\coset{gb})=\CircMatat(\coset{g})\CircMat{\hat{a}}(\coset{b}).
\]
As a consequence, the skew-constacyclic codes $\vva(\lideal{\coset{g}})$ and $\vv_{\hat{a}}(\lideal{\coset{gb}})$ are scale-equivalent,
that is, they differ only by rescaling each codeword coordinate with a fixed nonzero constant.
In particular, the codes have the same Hamming weight enumerator and Hamming distance.
\end{theo}
%%%%%%%%%%%%%%%%%%%%%%%%%%%%%
\begin{proof}
The first statement is due to Proposition~\ref{P-RightDiv}.
As for the circulants, we have trivially $b\rmid (x^n-\hat{a})$ and $\gamma(\hat{a},b)=a$.
Thus Theorem~\ref{T-ProdCirc} yields the desired identity.
The scale-equivalence follows from the fact that $\CircMat{\hat{a}}(\coset{b})$ is a non-singular diagonal matrix.
\end{proof}

%%%%%%%%%%%%%%%%%%%%%%%%%%%%%%
\begin{exa}\label{E-MonomEquiv}
Consider the situation of Example~\ref{E-RightDiv}(1); hence the map~$\vartheta$ from~\eqref{e-vartheta} is surjective.
The above tells us that it suffices to study $\theta$-cyclic codes, and thus the right divisors of~$x^n-1$, because each
$(\theta,a)$-constacyclic code is scale-equivalent to a $\theta$-cyclic one.
\end{exa}
%%%%%%%%%%%%%%%%%%%%%%%%%%%%%%

We return now to general circulants and show that if~$g$ is a right divisor of~$x^n-a$ then the transpose of
$\CircMatat(\coset{g})$ is a circulant, see~(1) below.
While this is an interesting result by itself, for us the version in~(2) relating the transpose to a different circulant is more powerful.
This is so because the polynomial $a\widehat{g}^{r}$ appearing in~(2) is a right divisor of $x^n-\theta^k(c^{-1})$, see Corollary~\ref{C-hhataghat}(b),
while~$g^\#$ in~(1) is not a right divisor of $x^n-c^{-1}$ (not even in the classical commutative case and with $a=c=1$).
As for Part~(2) below note that left multiplication of $\CircMatat(\coset{g})$ by $\CircMat{c}(\coset{x^k})$ is simply a reordering
and rescaling of the rows of~$\CircMatat(\coset{g})$; see the proof of Theorem~\ref{T-ProdCirc}.

%%%%%%%%%%%%%%%%%%%%%%%%%%%%%%%%
\begin{theo}\label{T-CircTrans}
Let $x^n-a=hg$, where $\deg(h)=k$, and let~$c=\gamma(a,g)$.
As in Corollary~\ref{C-hhataghat} let $\widehat{g}^{r}=\rho_r(\theta^n(g))$ and $\widehat{h}^{^l}=\rho_l(\theta^{-n}(h))$.
Then
\begin{arabiclist}
\item $M_a^{\theta}(\coset{g})\T=M_{c^{-1}}^{\theta}(\coset{g^\#})$, where $g^\#=a\widehat{g}^{r}x^k-cg_0(x^n-c^{-1})$,
\item $\CircMat{c}(\coset{x^k})\CircMatat(\coset{g})=\CircMat{\theta^k(c^{-1})}(\coset{a\widehat{g}^{r}})\T$,
\item $\CircMat{\theta^{k-n}(c^{-1})}(\coset{x^{n-k}})\CircMat{a^{-1}}(\coset{\widehat{h}^{^l}})=\CircMat{c}(\coset{a^{-1}h})\T$.
\end{arabiclist}
\end{theo}
%%%%%%%%%%%%%%%%%%%%%%%%%%%%%
\begin{proof}
(1) Write $g=\sum_{i=0}^{n-k}g_ix^i$ and set $g_i=0$ for $i=n-k+1,\ldots,n-1$.
Due to~\eqref{e-Mij} we have $\CircMatat(\coset{g})=(M_{ij})_{i,j=0,\ldots,n-1}$, where
\begin{equation}\label{e-Mgij}
    M_{ij}=\left\{\begin{array}{ll} \theta^i(g_{j-i}),&\text{if }i\leq j,\\[.7ex]
                           \theta^j(a)\theta^i(g_{n+j-i}),&\text{if }i>j.\end{array}\right.
\end{equation}
On the other hand, $\widehat{g}^{r}=\rho_r(\theta^n(g))=\sum_{i=0}^{n-k}\theta^{i+k}(g_{n-k-i})x^i$, and thus
$a\widehat{g}^{r}x^k=\sum_{i=k}^na\theta^i(g_{n-i})x^i$.
Using that $cg_0=a\theta^n(g_0)$, this leads to
\[
  g^\#= \sum_{i=0}^{n-1}s_ix^i,\ \text{ where }s_0=g_0\text{ and }s_i=a\theta^i(g_{n-i})\text{ for }i>0.
\]
Note that $s_i=0$ for $i=1,\ldots,k-1$.
By~\eqref{e-Mij}, $\CircMat{c^{-1}}(\coset{g^\#})=(P_{ij})_{i,j=0,\ldots,n-1}$, where
\[
   P_{ij}=\left\{\begin{array}{ll} \theta^i(s_{j-i})=\theta^i(a)\theta^j(g_{n-j+i}),&\text{ if }i< j,\\[.7ex]
                             \theta^i(s_0)=\theta^i(g_0),&\text{ if }i=j,\\[.7ex]
                             \theta^j(c^{-1})\theta^i(s_{n+j-i})=\theta^j(c^{-1})\theta^i(a)\theta^{n+j}(g_{i-j}),&\text{ if }i>j.
   \end{array}\right.
\]
This shows immediately that $P_{ij}=M_{ji}$ for all $i\leq j$.
The remaining case, that is, $P_{ij}=M_{ji}$ for $i>j$, is equivalent to the identities $g_t=c^{-1}\theta^t(a)\theta^n(g_t)$
for all $t:=i-j>0$.
But the latter have been established in Corollary~\ref{C-Coeffg}.
\\
(2) On the one hand, $\CircMat{c}(\coset{x^k})\CircMatat(\coset{g})=\CircMatat(\coset{x^kg})$ due to Theorem~\ref{T-ProdCirc}.
On the other hand,  for $\CircMat{\theta^k(c^{-1})}(\coset{a\widehat{g}^{r}})\T$ we may use part~(1) because~$a\widehat{g}^{r}$ is a right divisor of
$x^n-\theta^k(c^{-1})$ due to Corollary~\ref{C-hhataghat}(b).
Thus $\CircMat{\theta^k(c^{-1})}(\coset{a\widehat{g}^{r}})\T=\CircMat{b^{-1}}(\coset{(a\widehat{g}^{r})^\#})$, where
$b=\gamma\big(\theta^k(c^{-1}),a\widehat{g}^{r}\big)$
and~$(a\widehat{g}^{r})^\#$ is according to~(1).
The constant coefficient of $a\widehat{g}^{r}$ is $a\theta^k(g_{n-k})$ and hence
\begin{equation}\label{e-bainv}
  b=\gamma\big(\theta^k(c^{-1}),a\widehat{g}^{r}\big)=\theta^k(c^{-1})a^{-1}\theta^k(g_{n-k}^{-1})\theta^n(a)\theta^{n+k}(g_{n-k})=a^{-1},
\end{equation}
where the last step follows from the fact that the product of the last three factors is~$\theta^k(c)$ due to Corollary~\ref{C-Coeffg}.
All of this shows that $\CircMat{\theta^k(c^{-1})}(\coset{a\widehat{g}^{r}})\T=\CircMat{a}(\coset{(a\widehat{g}^{r})^\#})$, and it remains to prove that
$\coset{(a\widehat{g}^{r})^\#}=\coset{x^kg}$ in~$\cS_a$.
By definition, $\coset{(a\widehat{g}^{r})^\#}=\coset{\theta^k(c^{-1})\rho_r(\theta^n(a\widehat{g}^{r}))x^k}$.
Making use of Proposition~\ref{P-basics}(d),(f),(h) we compute
\begin{align*}
  \rho_r(\theta^n(a\widehat{g}^{r}))x^k&=\rho_r\big(\theta^n(a)\rho_r(\theta^{2n}(g))\big)x^k
      =\rho_r\circ\rho_r(\theta^{2n}(g))\theta^{k-n}(\theta^n(a))x^k\\
     &=\theta^{k-n}(\theta^{2n}(g))\theta^k(a)x^k=x^k\theta^n(g)a.
\end{align*}
Now~\eqref{e-thetahg} leads to $\theta^k(c^{-1})\rho_r(\theta^n(a\widehat{g}^{r}))x^k=x^kc^{-1}\theta^n(g)a=x^kg$, as desired.
\\
(3) follows from~(2): first $\widehat{h}^{^l}$ is a right divisor of~$x^n-a^{-1}$ due to Corollary~\ref{C-hhataghat}(c);
secondly $\gamma(a^{-1},\widehat{h}^{^l})=a^{-1}\big(\widehat{h}^{^l}_0\big)^{-1}\theta^n(\widehat{h}^{^l}_0)=\theta^{k-n}(c^{-1})$ due to Corollary~\ref{C-Coeffg} and because
$\widehat{h}^{^l}_0=\theta^{-n}(h_k)$; and finally
$a^{-1}\widehat{\widehat{h}^{^l}}^r=a^{-1}\rho_r(\theta^n(\rho_l(\theta^{-n}(h))))=a^{-1}h$, as desired.
\end{proof}

Theorems~\ref{T-ProdCirc} and~\ref{T-CircTrans}, true for right divisors~$g$ of~$x^n-a$, do not hold for more general polynomials.

%%%%%%%%%%%%%%%%%%%%%%%%%%%%%%%
\begin{exa}\label{E-ProdRule}
Let $\cR=\F_8[x;\theta]$, where~$\theta$ is the Frobenius homomorphism, thus $\theta(\lambda)=\lambda^2$ for all
$\lambda\in\F_8$. Let~$\alpha\in\F_8^*$ be the primitive element satisfying $\alpha^3=\alpha+1$.
Consider the polynomial $f:=x^5-\alpha^2$, hence $n=5$ and $a=\alpha^2$.
Then $h:=\alpha^6+x+\alpha^2x^2+\alpha^6x^3+x^4$ is a left divisor of~$f$, but not a right divisor.
In this case $\CircMatat(\coset{h})$ is in $\text{GL}_5(\F_8)$, and one can easily check that
$\CircMatat(\coset{xh})\CircMatat(\coset{h})^{-1}$ is not a circulant of the form $\CircMat{b}(\coset{s})$ for any~$s\in\cR$
and any~$b\in\F_8^*$.
This means that there is no identity of the form $\CircMatat(\coset{xh})=\CircMat{b}(\coset{s})\CircMatat(\coset{h})$, illustrating that
Theorem~\ref{T-ProdCirc} does not generalize.
Moreover, the transpose $\CircMatat(\coset{h})\T$ is not a circulant either.
\end{exa}
%%%%%%%%%%%%%%%%%%%%%%%%%%%%%%

%%%%%%%%%%%%%%%%%%%%%%%%%%%%%%%
\begin{theo}\label{T-MgMh}
Let $x^n-a=hg$, and as in Corollary~\ref{C-hhataghat} let $\widehat{h}^{^l}=\rho_l(\theta^{-n}(h))$. Then
\[
   \CircMatat(\coset{g})\CircMat{c}(\coset{a^{-1}h})= \CircMatat(\coset{g})\CircMat{a^{-1}}(\coset{\widehat{h}^{^l}})\T=0,\text{ where }
   c=\gamma(a,g).
\]
\end{theo}
%%%%%%%%%%%%%%%%%%%%%%%%%%%%%%%
\begin{proof}
For the first product we aim at using Theorem~\ref{T-ProdCirc} and thus need to check the requirements.
By Theorem~\ref{T-factors}(2) the polynomial $a^{-1}h$ is a right divisor of~$x^n-c$.
Moreover, $\gamma(c,a^{-1}h)=c(a^{-1}h_0)^{-1}\theta^n(a^{-1}h_0)=cah_0^{-1}\theta^n(h_0)\theta^n(a^{-1})=a$ by Corollary~\ref{C-Coeffg}.
Hence we may use Theorem~\ref{T-ProdCirc} and this yields $\CircMatat(\coset{g})\CircMat{c}(\coset{a^{-1}h})=\CircMat{c}(\coset{ga^{-1}h})$.
But the last matrix is zero because $\coset{ga^{-1}h}=\coset{0}$ in~$\cS_c$ due to Corollary~\ref{C-hhataghat}(a).
The rest follows from Theorem~\ref{T-CircTrans}(3).
\end{proof}

%%%%%%%%%%%%%%%%%%%%%%%%%%%%%%%%%
\section{The lattices of skew-constacyclic codes}\label{S-lattice}
Let $\cR:=\F[x;\theta]$ for some fixed $\theta\in\Aut(\F)$.
The previous sections lead to the following result, which was first presented and proven in a different form by
Boucher/Ulmer in \cite[Thm.~8]{BoUl09a} and~\cite[Thm.~1]{BoUl11}.
%%%%%%%%%%%%%%%%%%%%%%%%%%%%
\begin{theo}\label{T-CodesDual}
Let~$a\in\F^*$ and $\cC\subseteq\F^n$ be a $(\theta,a)$-constacyclic code.
Then there exists a unique monic polynomial $g\in\cR$ such that $x^n-a=hg$ for some $h\in\cR$ and
$\cC=\im\CircMatat(\coset{g})=\vva(\lideal{\coset{g}})$.
In this case $\cC^{\perp}$ is $(\theta,a^{-1})$-constacyclic and
$\cC^{\perp}=\im\CircMat{a^{-1}}(\coset{\widehat{h}^{^l}})=\vv_{a^{-1}}(\lideal{\coset{\widehat{h}^{^l}}})$, where $\widehat{h}^{^l}=\rho_l(\theta^{-n}(h))$.
\end{theo}
%%%%%%%%%%%%%%%%%%%%%%%%%%%%
\begin{proof}
The first part about~$\cC$ is in Theorem~\ref{T-tacyclic} and  Proposition~\ref{P-basicCM}.
As for the dual code, note first that
$\rank(\CircMatat(\coset{g}))=n-\deg(g)=\deg(h)=\deg(\widehat{h}^{^l})=n-\rank(\CircMat{a^{-1}}(\coset{\widehat{h}^{^l}}))$.
Since Theorem~\ref{T-MgMh} yields $\CircMatat(\coset{g})\CircMat{a^{-1}}(\coset{\widehat{h}^{^l}})\T=0$ we conclude that
$\im\CircMatat(\coset{g})$ and $\im \CircMat{a^{-1}}(\coset{\widehat{h}^{^l}})$ are mutually dual codes.
\end{proof}

Now we recover \cite[Prop.~13]{BoUl09a} about self-dual codes (see also \cite[Prop.~5]{BoUl14a}).
%%%%%%%%%%%%%%%%%%%%%%%%%%%%%
\begin{cor}\label{C-SelfDual}
If there exists a self-dual $(\theta,a)$-constacyclic code in~$\F^n$, then $n$ is even and $a=\pm1$.
\end{cor}
%%%%%%%%%%%%%%%%%%%%%%%%%%%%

We are now in a position to formulate the interplay between right divisors of~$x^n-a$ and the associated codes as well as
their duals in terms of lattice (anti-)isomorphisms.
For $a\in\F^*$ define the sets
\begin{align*}
   \cD_a&:=\{g\in\cR\mid g\rmid (x^n-a),\, g\text{ monic}\},\\[.6ex]
   \cI_a&:=\{I\subseteq\cS_a\mid I\text{ is a submodule of }\cS_a\},\\[.6ex]
   \cT_a&:=\{\cC\subseteq\F^n\mid \cC\text{ is $(\theta,a)$-constacyclic}\}.
   %\cM_a&:=\{\CircMatat(\coset{g})\mid g\in\cD_a\}.
\end{align*}
Clearly, $(\cD_a,\,\rmid\,),\ (\cI_a,\,\subseteq\,),\ (\cT_a,\,\subseteq\,)$ are lattices.
Consider the maps
\begin{equation}\label{e-maps1}
\begin{array}{ccccccc}
   \cD_a &\stackrel{\sigma_a}{\longrightarrow} &\cT_a&\stackrel{\ppa}{\longrightarrow}& \cI_a\\[.5ex]
     g   &\longmapsto& \im \CircMatat(\coset{g})&\longmapsto& \ppa\big(\im \CircMatat(\coset{g})\big)
\end{array}
\end{equation}
Because of Corollary~\ref{C-CircSpaces}(a) and Theorem~\ref{T-tacyclic}, the map~$\sigma_a$ is a lattice anti-isomorphism,
while~$\ppa$ is a lattice isomorphism thanks to Proposition~\ref{P-basicCM}.

We now turn to the dual situation.
Let $x^n-a=hg$ with monic polynomials $g,h\in\cR$.

%%%%%%%%%%%%%%%%%%%%%%%%%%%
\begin{theo}\label{T-DualCode}
Define the map $\delta_a:\cD_a\longrightarrow \cD_{a^{-1}},\ g\longmapsto \theta^{-\deg(g)}(-a^{-1}g_0)\widehat{h}^{^l}$, where~$g_0$
is the constant coefficient of~$g$ and, as before, set $\widehat{h}^{^l}:=\rho_l\big(\theta^{-n}(h)\big)$.
Moreover, define $\tau_a:\cT_a\longrightarrow\cT_{a^{-1}},\ \cC\longmapsto \cC^{\perp}$, and let~$\sigma_a$ be
as in~\eqref{e-maps1}.
Consider the diagram
\[
\begin{array}{l}
   \begin{xy}
   (0,18)*+{\cD_a}="a"; (18,18)*+{\;\cD_{a^{-1}}}="b";%
   %(15,0)*+{\;V}="d";%
   (0,0)*+{\cT_a}="c"; (18,0)*+{\;\cT_{a^{-1}}}="d";
   {\ar "a";"b"}?*!/_2mm/{\delta_a\;};
   {\ar "c";"d"}?*!/_2mm/{\tau_a\;};
   {\ar "b";"d"}?*!/_4mm/{\,\sigma_{a^{-1}}};%
   {\ar "a";"c"}?*!/_4mm/{\mbox{}\!\!\!\!\sigma_a};%
   \end{xy}
\end{array}
\]
Then all maps are lattice anti-isomorphisms and the diagram commutes.
In other words, if $\cC=\im\CircMatat(\coset{g})$ for some $g\in\cD_a$, then
$\cC^{\perp}=\im\CircMat{a^{-1}}(\coset{\delta_a(g)})=\im\CircMat{a^{-1}}(\coset{\widehat{h}^{^l}})$.
\end{theo}
%%%%%%%%%%%%%%%%%%%%%%%%%%%%%
\begin{proof}
First of all,~$\delta_a(g)$ is indeed a right divisor of $x^n-a^{-1}$ thanks to Corollary~\ref{C-hhataghat}(c), and it is monic because
the leading coefficient of~$\widehat{h}^{^l}$ is $\theta^{-\deg(g)}(-ag_0^{-1})$, as one can easily verify.
Next, Theorem~\ref{T-CodesDual} yields that the diagram commutes.
This in turn implies that~$\delta_a$ is a lattice anti-isomorphism because~$\sigma_a,\,\tau_a,\,\sigma_{a^{-1}}$ are.
\end{proof}

Now we can present the dual lattices to those in Example~\ref{E-FactorsXna}.
%%%%%%%%%%%%%%%%%%%%%%%%%%%%%%%
\begin{exa}\label{E-FactorsXnadual}
Consider again the field~$\F_8=\F_2[\alpha]$, where $\alpha^3=\alpha+1$, and let~$\theta$ be the Frobenius homomorphism on~$\F_8$.
In Example~\ref{E-FactorsXna} we presented all monic right divisors of $x^7+\alpha$.
Using the map~$\delta_{\alpha}$ we obtain all right divisors of $x^7+\alpha^{-1}=x^7+\alpha^6$.
Setting $\tilde{h}^{(i)}:=\delta_{\alpha}(g^{(i)})$ for $i=0,\ldots,7$, we obtain
\begin{align*}
  &\tilde{h}^{(0)}=x^7+\alpha^6,\quad \tilde{h}^{(1)}=x^6+\alpha^3 x^5+\alpha x^4+x^3+\alpha^3 x^2+\alpha x+1,\\[.6ex]
  &\tilde{h}^{(2)}=x^4+\alpha^2x^2+x+\alpha^6,\quad \tilde{h}^{(3)}=x^4+\alpha^6x^3+\alpha^2x^2+\alpha^6,\quad
    \tilde{h}^{(4)}=x^3+\alpha x+1\\[.6ex]
  &\tilde{h}^{(5)}=x^3+\alpha^3 x^2+1,\quad \tilde{h}^{(6)}=x+\alpha^6 ,\quad \tilde{h}^{(7)}=1.
\end{align*}
From the above we know that $(\cC^{(i)})^{\perp}=\sigma_{\alpha}(\tilde{h}^{(i)})$, and thus we obtain the lattices given in Figure~\ref{F-Lattices2}.
They are dual to those in Figure~\ref{F-Lattices1}.
\begin{figure}[ht]
\centering
  \mbox{}\hspace*{1cm}
  \includegraphics[height=5cm]{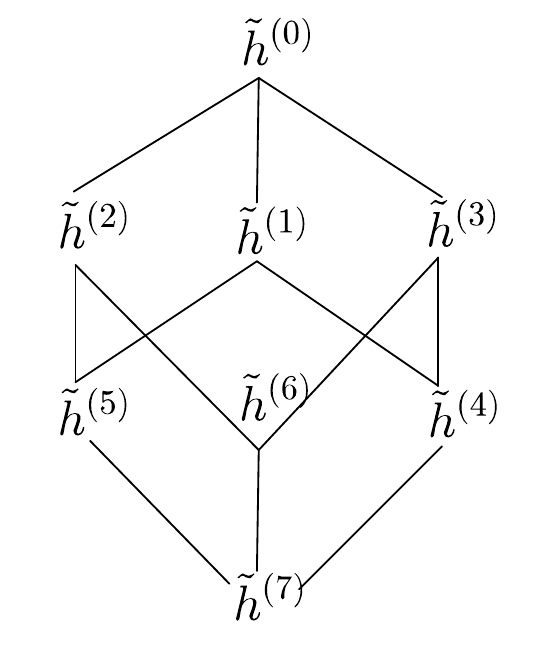}
         \hspace*{2cm}\includegraphics[height=5cm]{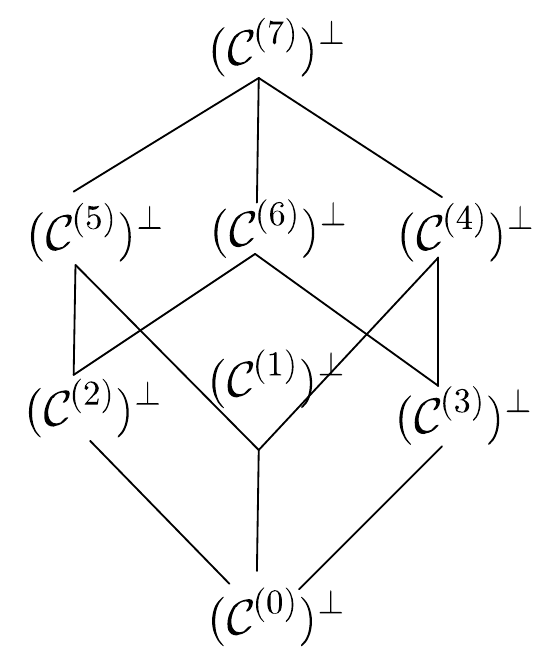}
    \caption{Lattice of monic right divisors of $x^7+\alpha^{-1}$ and the corresponding codes}
    \label{F-Lattices2}
\end{figure}
\end{exa}
%%%%%%%%%%%%%%%%%%%%%%%%%%%%%%%%%%

We now turn to the notion of a check polynomial for skew-constacyclic codes.
%%%%%%%%%%%%%%%%%%%%%%%%%%
\begin{prop}\label{P-ParityCheck}
Let $x^n-a=hg$ and $c=\gamma(a,g)$.
Then the map
\[
   \psi:\cS_a\longrightarrow \cS_{\theta^{-n}(c)},\ \coset{f}\longmapsto \coset{f\theta^{-n}(h)}
\]
is a well-defined $\cR$-module homomorphism with $\ker\psi=\lideal{\coset{g}}$.
\end{prop}
%%%%%%%%%%%%%%%%%%%%%%%%%
\begin{proof}
Well-definedness and the containment $\ker\psi\supseteq\lideal{\coset{g}}$ follow from Theorem~\ref{T-factors}(3),
and $\cR$-linearity is clear.
For $\ker\psi\subseteq\lideal{\coset{g}}$ note that $f\theta^{-n}(h)=t(x^n-\theta^{-n}(c))$ for some $t\in\cR$ implies
$f\theta^{-n}(h)=tg\theta^{-n}(h)$ and thus $f\in\lideal{g}$ by right cancellation in~$\cR$.
\end{proof}

The last result justifies to call $\theta^{-n}(h)$ the \emph{check polynomial} of the code
$\cC=\vva(\lideal{\coset{g}})$.
The only thing to keep in mind that the check equation is carried out modulo~$x^n-\theta^{-n}(c)$.
This generalizes \cite[Lem.~8]{BoUl09} (see also~\cite[Thm.~2.1(iii)]{GSF13}), where a central polynomial $x^n-1$
is considered.
In that case $\theta^n$ is the identity on~$\cR$ and thus $\theta^{-n}(h)=h$.
In particular, all of this generalizes the classical commutative case where~$h$ is the check polynomial of~$\cC$~\cite[Ch.~7, \S4]{MS77}.

\medskip
We close with a brief summary of the central case.
The results bear some resemblance with those obtained for cyclic convolutional codes in~\cite{GS04}; see especially Theorem~7.5 therein.
The last part of~(4) appears already in~\cite[Cor.~1]{Mat10} by Matsuoka,
where even skew-polynomial rings over arbitrary finite rings are considered.

%%%%%%%%%%%%%%%%%%%%%%%%%%%%%%%
\begin{theo}\label{T-twosided}
Let~$n$ be such that $\theta^n=\mbox{\rm id}_{\cR}$ and consider $x^n-a$ for some $a\in\Fix_{\F}(\theta)$, hence $x^n-a$ is central.
Suppose $x^n-a=hg$.
Then
\begin{arabiclist}
\item $\CircMatat$ induces an injective ring homomorphism from~$\cS_a$ into $\F^{n\times n}$.
\item $x^n-a=gh$.
\item $\CircMatat(\coset{g})\CircMatat(\coset{h})=\CircMatat(\coset{h})\CircMatat(\coset{g})=0$.
\item We have left $\cR$-module homomorphisms
      \[
         \psi_h:\cS_a\longrightarrow\cS_a,\ \coset{f}\longmapsto \coset{fh}\ \text{ and }\
         \psi_g:\cS_a\longrightarrow\cS_a,\ \coset{f}\longmapsto \coset{fg}.
      \]
      Moreover, $\ker\psi_h=\lideal{\coset{g}}=\mbox{\rm ann}_l(\rideal{\coset{h}})$,
      the left annihilator of the right ideal generated by~$h$.
      In the same way, $\ker\psi_g=\lideal{\coset{h}}=\mbox{\rm ann}_l(\rideal{\coset{g}})$.
      In this sense~$h$ is the check polynomial of the code $\cC=\vva(\lideal{\coset{g}})$.
\item We have right $\cR$-module homomorphisms
      \[
         \psi'_h:\cS_a\longrightarrow\cS_a,\ \coset{f}\longmapsto \coset{hf}\ \text{ and }\
         \psi'_g:\cS_a\longrightarrow\cS_a,\ \coset{f}\longmapsto \coset{gf},
      \]
      and $\ker\psi'_h=\rideal{\coset{g}}=\mbox{\rm ann}_r(\lideal{\coset{h}})$,
      the right annihilator of the left ideal generated by~$h$, and
      $\ker\psi'_g=\rideal{\coset{h}}=\mbox{\rm ann}_r(\lideal{\coset{g}})$.
\item Let $\cC=\vva(\lideal{\coset{g}})$ and $h=\sum_{i=0}^k h_i x^i$.
      Then
      \[
         \cC^{\perp}=\vv_{a^{-1}}(\lideal{\coset{\rho_l(h)}}),
         \text{ where }\rho_l(h)=h_k+\theta(h_{k-1})x+\ldots+\theta^k(h_0)x^k.
      \]
\end{arabiclist}
\end{theo}
%%%%%%%%%%%%%%%%%%%%%%%%%%%%%%
One may regard~(5) and~(6) as the counterpart to~(4) in terms of ideals.

\begin{proof}
(1) is in Theorem~\ref{T-centralCirc}.
(2) follows from Theorem~\ref{T-factors} because $\gamma(a,g)=a$ for all right divisors~$g$ of~$x^n-a$.
(3) is a consequence of~(1) and~(2).
(4) is a special case of Proposition~\ref{P-ParityCheck}, and~(5) follows by symmetry.
(6) is a special case of Theorem~\ref{T-CodesDual}.
\end{proof}

In this context it is worth pointing out that if $x^n-a$ is central and $x^n-a=hg$ then~$g$ and~$h$ need not even be two-sided:
for instance, in $\F_4[x;\theta]$ with~$\theta$ being the Frobenius homomorphism, we have the identity
$x^4-1=(x^2+\alpha x+\alpha^2)(x^2+\alpha x+\alpha)$, and neither factor is two-sided.
Furthermore, if $x^n-a$ is a product of three or more factors, the factors do not commute arbitrarily.
This can be seen with $x^6-1=(x+1)(\alpha^2x^2+1)(\alpha x^3+\alpha x^2+x+1)\neq(\alpha x^3+\alpha x^2+x+1)(\alpha^2x^2+1)(x+1)$
in~$\F_4[x;\theta]$.
It is well known that every two-sided element can be factored into a product of two-sided maximal elements, and in this case
the factors commute~\cite[Sec.~1.2]{Jac96}.
Further information about the case where $a=1$ and $x^n-1$ is central can be found in~\cite{GSF13}.

%%%%%%%%%%%%%%%%%%%%%%%%%%%%%%%%%%%%%%%%%%%%%%
\bibliographystyle{abbrv}

%%%%%%%%%%%%%%%%%%%%%%%%%%%%%%%%%%%%%%%%%%%%%%

\end{document}